\documentclass[sigchi]{acmart}
\pdfoutput=1 
\usepackage[utf8]{inputenc}
\usepackage{subcaption}
\usepackage[font=bf]{caption}
\usepackage{amsfonts}
\usepackage{booktabs}
\usepackage{amsmath}
\usepackage{color}
\usepackage{balance}
\usepackage{flushend}
\usepackage{makecell}
\usepackage{bbding}
\usepackage{pifont}
\usepackage{wasysym}
\usepackage{amssymb}
\usepackage{multirow}

\setcopyright{none}

\copyrightyear{2019} 
\acmYear{2019} 
\acmConference[CHI 2019]{CHI Conference on Human Factors in Computing Systems Proceedings}{May 4--9, 2019}{Glasgow, Scotland Uk}
\acmBooktitle{CHI Conference on Human Factors in Computing Systems Proceedings (CHI 2019), May 4--9, 2019, Glasgow, Scotland Uk}
\acmPrice{15.00}
\acmDOI{10.1145/3290605.3300248}
\acmISBN{978-1-4503-5970-2/19/05}

\usepackage{enumitem}
\newcommand{\squishlist}{\begin{itemize}[itemsep=1pt,parsep=2pt,topsep=3pt,partopsep=0pt,leftmargin=0em, itemindent=1em,labelwidth=1em,labelsep=0.5em]}
\newcommand{\squishend}{\end{itemize}}
\newcommand{\squishenum}{\begin{enumerate}[itemsep=1pt,parsep=2pt,topsep=3pt,partopsep=0pt,leftmargin=0em,listparindent=1.5em,labelwidth=1em,labelsep=0.5em]}
\newcommand{\squishsubenum}{\begin{enumerate}[itemsep=1pt,parsep=2pt,topsep=0pt,partopsep=0pt,leftmargin=0em,listparindent=1.5em,labelwidth=1em,labelsep=0.5em]}
\newcommand{\squishenumend}{\end{enumerate}}

\newcommand{\red}[1]{#1}

\newcommand{\name}{MilliSonic}

\begin{document}
\title{MilliSonic: Pushing the Limits of Acoustic Motion Tracking}

\author{Anran Wang and Shyamnath Gollakota}
\affiliation{%
  \institution{University of Washington}
}
\email{{anranw, gshyam}@cs.washington.edu}


\begin{abstract}
 Recent years have seen interest in device tracking and localization using acoustic signals. State-of-the-art acoustic motion tracking systems however do not achieve millimeter accuracy and require large separation between microphones and speakers, and as a result, do not meet the requirements for many VR/AR applications. Further, tracking multiple concurrent acoustic transmissions from VR devices today requires sacrificing accuracy or frame rate. We present MilliSonic, a novel system that pushes the limits of acoustic based motion tracking. Our core contribution is a novel localization algorithm that can provably achieve sub-millimeter 1D  tracking accuracy in the presence of multipath, while using only a single beacon with a small   4-microphone array. Further, \name\ enables concurrent tracking of upto four smartphones without reducing frame rate or accuracy.   Our evaluation shows that \name\ achieves 0.7mm median 1D  accuracy and a 2.6mm median 3D  accuracy for smartphones, which is  5x more accurate than state-of-the-art systems.  MilliSonic enables two previously infeasible interaction applications: a) 3D tracking of VR headsets  using the smartphone as a beacon and b) fine-grained 3D tracking  for the Google Cardboard VR system  using a small microphone array. 
\end{abstract}

\ccsdesc[500]{Human-centered computing~Pointing}
\ccsdesc[500]{Human-centered computing~Ubiquitous and mobile computing systems and tools}


\maketitle

\section{Introduction}

\begin{figure*}
\vskip -0.05in
    \centering
    \begin{subfigure}[b]{0.25\textwidth}
        \includegraphics[width=0.8\textwidth]{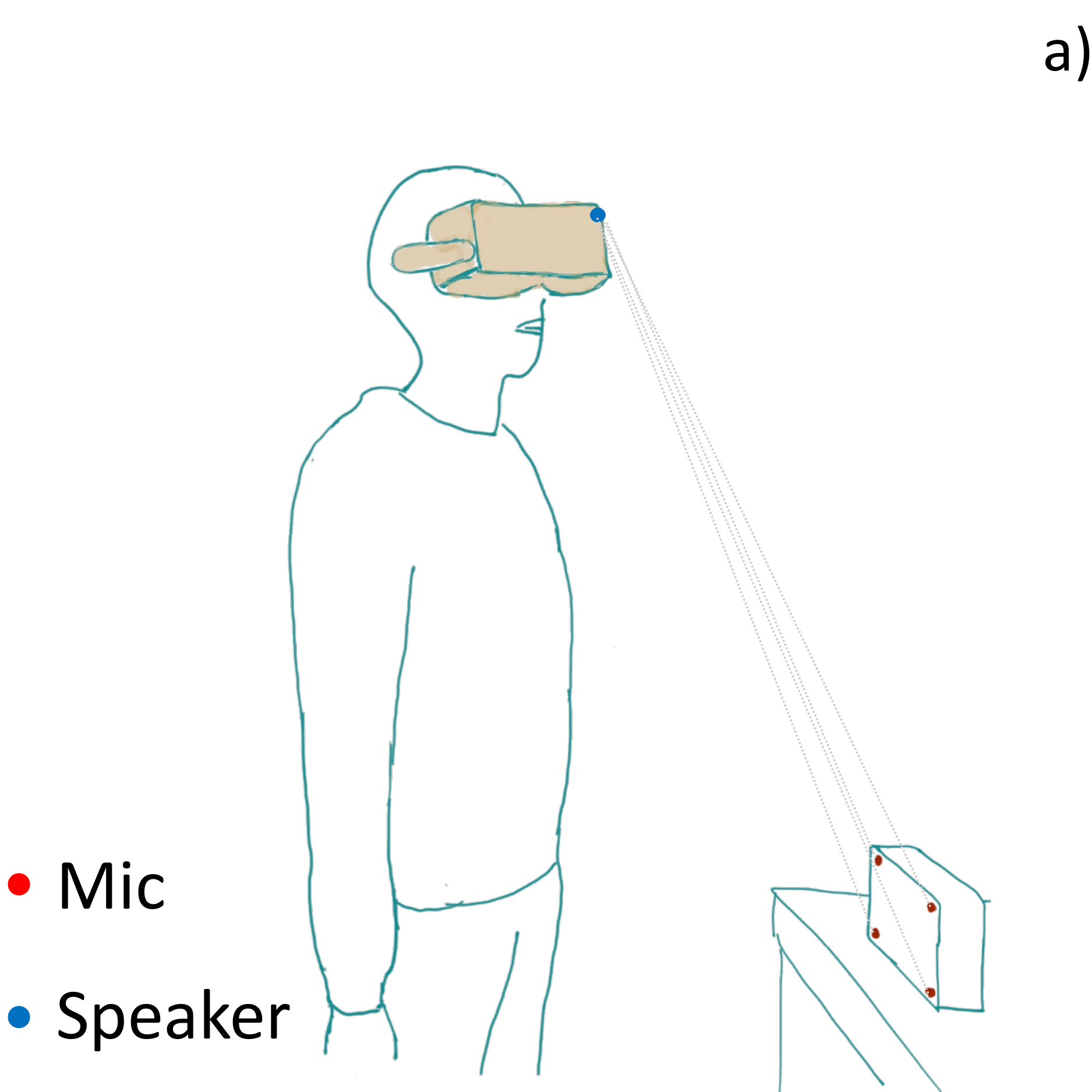}
    \end{subfigure}
    \hspace{3em}
    \begin{subfigure}[b]{0.25\textwidth}
        \includegraphics[width=0.8\linewidth]{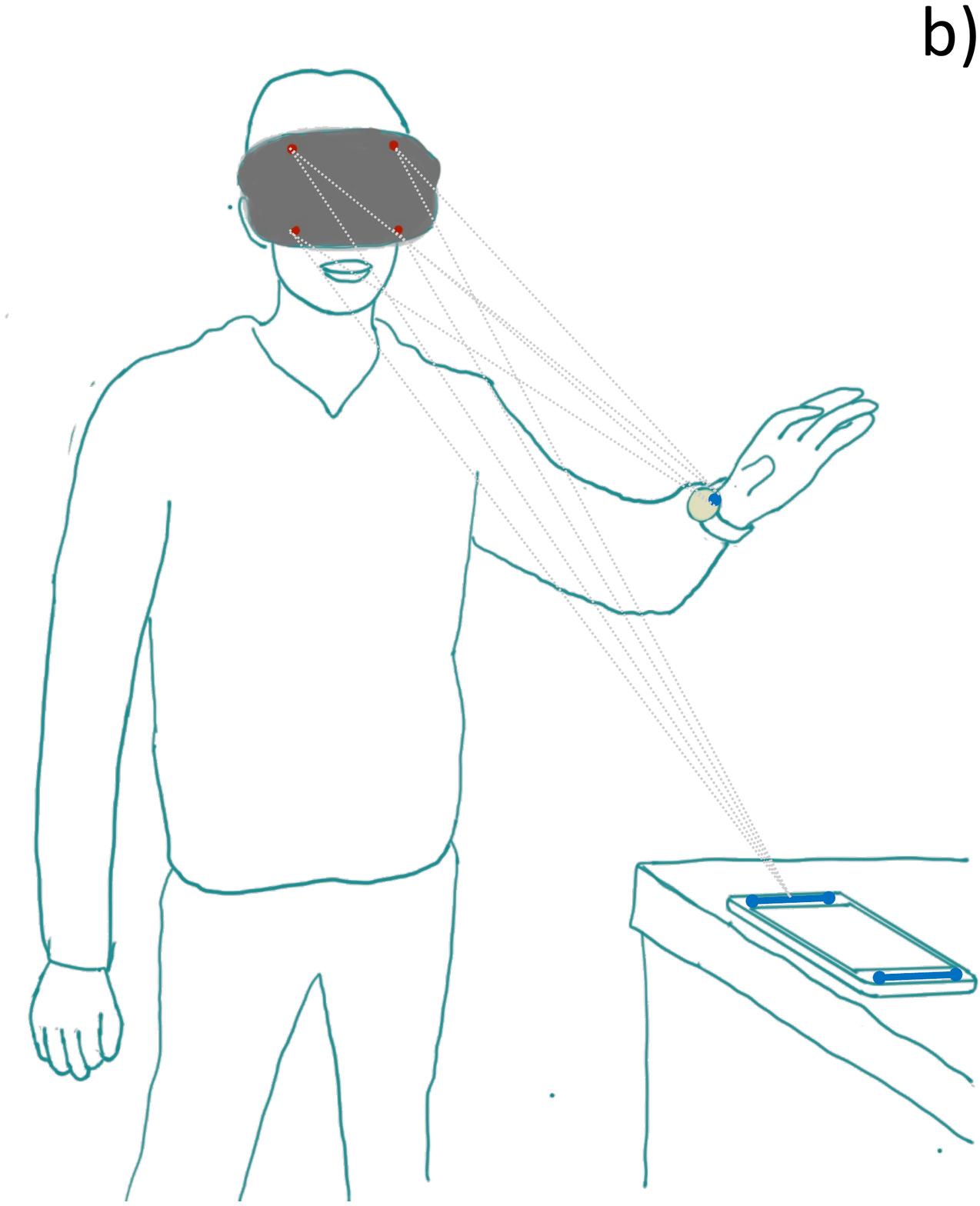}
    \end{subfigure}
    \hspace{3em}
    \begin{subfigure}[b]{0.25\textwidth}
        \includegraphics[width=0.8\linewidth]{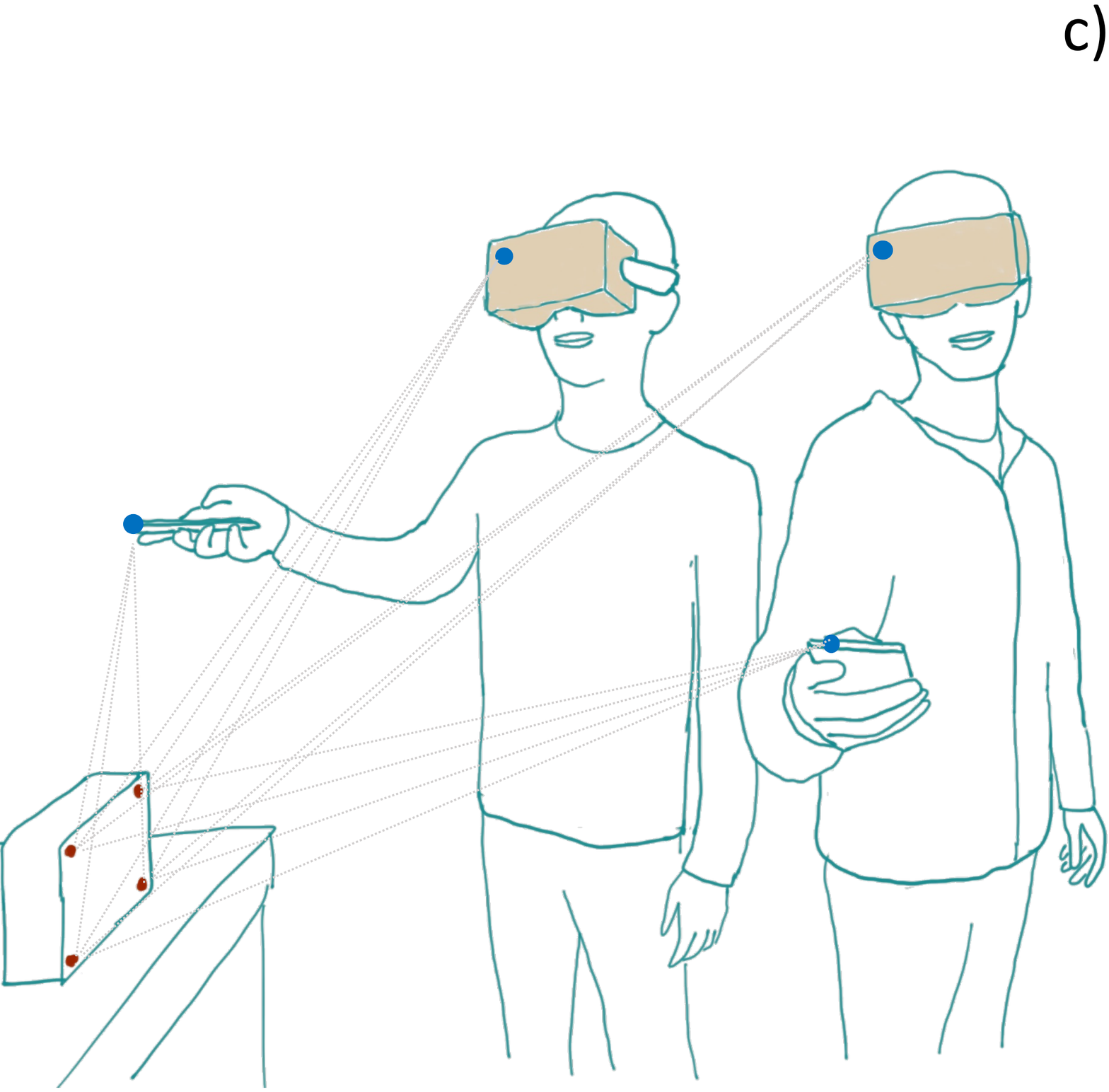}
    \end{subfigure}
    \vskip -0.15in
    \caption{Application scenarios: a) tracking a Google cardboard VR using a small microphone array; b) Tracking the 3D position of VR/AR headsets using a smartphone as a beacon. Using the transmissions from the smartwatch to then track it w.r.t. the headset; c) Concurrently tracking multiple devices with a single microphone array at a high per-device frame rate.}
    \label{fig:illustration}
\end{figure*}


Device localization and motion tracking has been a  long-standing challenge in the research community. It is a key component in  Virtual Reality  and Augmented/Mixed Reality applications and enables novel human-computer interactions including gesture and skeletal tracking. Traditionally, specialized optical methods  such as lasers and infrared beacons have been used to localize  VR headsets and controllers. This includes commercial systems like the HTC Vive VR, Oculus Rift and Sony PlayStation VR~\cite{htcvive, oculusrift, playstationmove}. These optical tracking solutions, however, require separate \red{expensive} beacons to emit infrared signals and  transceivers to receive and process data. Existing devices like smartphones lack  these transceivers and hence are unsuitable for such techniques.

Acoustic-based localization and tracking methods have recently emerged as an attractive alternative to optical systems~\cite{peng2007beepbeep, zhang2012swordfight}.  Speakers and microphones, used for emitting and receiving acoustic signals, are cheap and easy to configure. Furthermore, commodity smartphones and smart watches already have built-in speakers and microphones, which makes acoustic tracking an excellent fit for such devices. As shown in Fig.~\ref{fig:illustration}(a), a simple microphone array could act as a beacon to enable 3D location tracking for the Google cardboard VR system. Conversely, instead of carrying around additional devices (e.g., HTC IR beacons) to enable tracking for VR headsets, one could  reuse existing smartphones as beacons to enable 3D localization and motion tracking.

State-of-the-art acoustic motion tracking systems~\cite{mao2016cat, zhang2017soundtrak} however do not adequately meet the requirements of VR/AR applications for three main reasons.
\squishlist
\item {\it Tracking accuracy.} Acoustic signals suffer from multi-path where the signal reflects off nearby surfaces before arriving at the receiver. Thus, existing 1D acoustic tracking accuracy is 5-10~mm~\cite{mao2016cat}, which is much worse than optical systems \red{and may cause motion sickness with prolonged use~\cite{Akiduki2003VisualvestibularCI}}.  

\item {\it Microphone/speaker separation.} 3D tracking requires triangulation  from multiple microphones/speakers, which when placed close to each other limits accuracy. Prior work that tracks smartphones uses multiple speakers separated by 90~cm~\cite{mao2016cat}, making them difficult to integrate  into VR/AR headsets. Conversely, using a 90~cm beacon for Google cardboard VR  is unwieldy and limits portability.
\item {\it Concurrency.} Tracking multiple headsets remains a challenge with existing designs. A na\"{\i}ve approach is to time multiplex the acoustic signals from each device. This however reduces  the frame rate linearly with the number of devices.
\squishend

We present \name, a novel system that pushes the limits of acoustic based motion tracking. Our core contribution is a novel localization algorithm that can achieve sub-millimeter 1D tracking accuracy in the presence of multipath, while using only a single beacon with a 4-microphone array.  To achieve this, like prior designs~\cite{mao2016cat, erdelyi2018sonoloc}, \name\ uses FMCW (frequency modulated continuous wave) acoustic transmissions where the frequency linearly increases with time.  Prior designs use FMCW to separate reflections arriving at different times by mapping time differences to frequency shifts. However, given the limited inaudible bandwidth on smartphones, the ability to differentiate between close-by paths using frequency shifts is limited, thus, limiting accuracy. Our algorithm instead leverages the phase of the FMCW reflections to perform tracking. We prove that this allows us to achieve sub-millimeter 1D tracking. These high 1D accuracies allow us to reduce the separation between microphones at the beacon and achieve millimeter-resolution 3D tracking and localization. Finally, we show that by have devices intentionally introduce different time delays to their FMCW signals, we can support concurrent acoustic transmissions from multiple devices, without reducing the accuracy or frame rate for each device. 

 We implement our design using speakers on Android smartphones including Samsung Galaxy S6, S7 and S9. We design $15cm\times 15cm$ and {$6cm\times 5.35cm$} 4-microphone arrays using commercial microphones and implement our real-time tracking algorithms on a Raspberry Pi 3 Model B+~\cite{raspberrypi}.  
 
This paper makes the following contributions.
\squishlist
  \item We show for the first time how to achieve sub-mm 1D tracking and localization accuracies using acoustic signals on smartphones, in the presence of multipath. To achieve this, we introduce algorithms that use the phase of FMCW signals to disambiguate between multiple paths.
    \item We enable multiple  smartphones to transmit concurrently using time-shifted FMCW acoustic signals and enable concurrent tracking  without sacrificing  accuracy or frame rate.
    \item We present experimental results that show that \name\ can achieve a median 1D   accuracy of 0.7~mm up to distances of 1~m from the smartphone.  The median 1D  accuracy is 1.7~mm  for distances  between 1 and 2~m. \name's median 3d  accuracy is around 2.6~mm.  Further, we can concurrently track up to four smartphones  at a per-device frame rate of 40 frames/sec without sacrificing accuracy.
    \item \red{Finally, we describe the limitations of our system and outline additional work required to more comprehensively evaluate the system in various use case scenarios.}
\squishend

\section{Application Scenarios}
\red{
MilliSonic enables three key  application scenarios.
\squishlist
  \item Current smartphone-based VR headsets (\textit{e.g.}, Google Cardboard) do not have 6DoF motion tracking capability. This is because of the lack of optical transceivers, which limits their usage. MilliSonic enables 6DoF motion tracking capability for smartphone-based VR headsets using only a cheap and small microphone array as a beacon,  without requiring any hardware modifications at the smartphone.
  \item Millisonic can transform the smartphone into a portable beacon for VR tracking. Specifically, instead of requiring the user to carry optical beacons for VR headsets to enable use in different environments, a smartphone can be used as a portable beacon.  To do this, manufacturer can integrate a cheap microphone array into the VR/AR headset. Using this microphone array, the VR headset can also track the motion of other acoustic-enabled devices such as smart watches.
  \item MilliSonic can support concurrent tracking of an unlimited number of microphone arrays (i.e., VR headsets) in the vicinity of a single speaker (i.e., a smartphone). Furthermore, it can also support up to four speakers (i.e., smartphone VR headsets) in the vicinity  of a microphone array without sacrificing accuracy or frame rate.
\squishend
}
\section{\name\ Design}
We first present background on existing FMCW tracking systems and show why they have a limited accuracy for acoustic tracking. We then present our algorithm that uses the FMCW phase to achieve sub-mm 1D tracking. We then describe how to perform 3D tracking using the 1D locations using multiple microphones. Finally, we address various practical issues. 


%



\subsection{FMCW Background}

Acoustic tracking is traditionally achieved by computing  the time-of-arrival of the transmitted signals at the microphones.  At its simplest, the transmitted signal is a sine wave, $x(t)=exp(-j2\pi ft)$ where $f$ is the wave frequency. A microphone at a distance of $d$ at the transmitter, has a time-of-arrival of \red{$d = t_d \times c$} where $c$ is the speed of sound. The received signal at this distance can now be written as, $y(t) = exp(-j2\pi t(t-t_d)$. Dividing by $x(t)$, we get $\hat{y}(t)=\exp(j2\pi ft_d)$. Thus, the phase of the received signal can be used to compute the time-of-arrival, $t_d$. In practice, however, multipath significantly distorts the received phase limiting accuracy. 

To combat multipath, prior work~\cite{mao2016cat, nandakumar2015contactless} uses  Frequency Modulated Continuous Wave (FMCW) chirps where as shown in Fig.~\ref{fig:fmcw} the frequency of the signal changes linearly with time. FMCW has good autocorrelation properties that allow the receiver to \red{differentiate} between multiple paths that each have a different time-of-arrival. {Further compared to  OFDM~\cite{fingerio} and other waveforms \cite{cazac}, FMCW has high spectral efficiency and is ease of demodulate.} Mathematically, the FMCW signal in Fig.~\ref{fig:fmcw} is,
 $x(t)=exp(-j2\pi (f_0+\frac{B}{2T}t)t) = 
 exp(-j2\pi (f_0t+\frac{B}{2T} t^2))$, where $f_0$, $B$ and $T$ are the initial frequency, bandwidth and duration of the FMCW chirp respectively. In the presence of multipath, the received signal can be written as, $y(t)=\sum_{i=1}^M A_iexp(-j2\pi(f_0(t-t_i)+\frac{B}{2T} (t^2+t_i^2-2tt_i)))$, where $A_i$ and $t_i=\frac{d_i(t)}{c}$ are the attenuation and time-of-flight of the $i$th path. Dividing this by $x(t)$ we get,
\begin{align}
\hat{y}(t)&=\sum_{i=1}^M A_iexp(-j2\pi (\frac{B}{T}t_it+f_0t_i-\frac{B}{2T}t_i^2)) \label{eq:fmcwdemod}
\end{align}
The above equation shows that multipath with different times-of-arrival fall into different frequencies. The receiver uses \red{Discrete Fourier Transform (DFT)} to find the first peak frequency bin, $f_{peak}$, that corresponds to the line-of-sight path to the transmitter. It then computes the distance to the receiver as, $d(t)=\frac{cf_{peak}}{B}$.

While this conventional FMCW processing is effective in disambiguating multiple paths that is separated by large distances, it has a limited accuracy when the multiple paths are close to each other. Specifically, the minimum distance resolution for FMCW is in the order of $\frac{c}{B}$ when the separation between frequencies is 1~Hz. Given that smartphones have a limited inaudible bandwidth of {7~kHz} between {17-24~kHz}, prior work cannot distinguish between paths that are close to the direct line-of-sight path and hence have a limited accuracy~\cite{mao2016cat, peng2007beepbeep}. Further, since DFT operations are performed over a whole chirp duration, it limits the frame rate of the system to $\frac{1}{T}$, where $T$ is the FMCW chirp duration.

\begin{figure}
\includegraphics[width=0.25\textwidth]{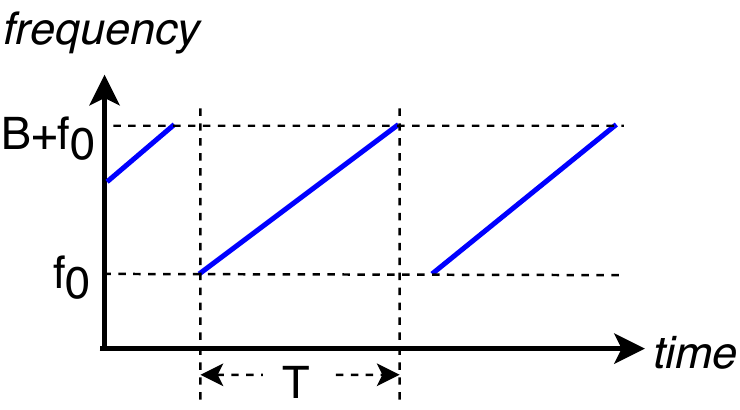}
\vskip -0.15in
\caption{FMCW signal structure.}
\vskip -0.2in
\label{fig:fmcw}
\end{figure}

\subsection{Sub-mm 1D tracking using FMCW phase}

We use phase of the FMCW signals to compute distance. Thus, instead of  using the first peak frequency of the FMCW signal in the frequency domain to estimate the time-of-arrival, our algorithms has two key steps: 1) we apply a dynamic narrow band-pass filter in the time-domain to filter out most multipath that has a distant time-of-arrival from the direct path. This leaves us only a small portion of residual indirect paths around the direct path. 2) We  then extract the distance information from the instantaneous FMCW phase. 

\red{
{\bf Intuition.}  We provide the intuition  for why FMCW phase provides a better accuracy than existing FMCW approaches.}

\red{{\it Traditional FMCW approaches.} Let us first understand the error in traditional peak estimation method for FMCW signals. Tracking error occurs when we have two paths that are without a single frequency bin.  Let us denote the time-of-arrival of the direct path as $t_1$ and its frequency in the demodulated FMCW signal as $\mathbf{f}_{t_1}$. An indirect path with a time-of-arrival of $t_2$ lies at frequency $\mathbf{f}_{t_2}$ in the demodulated signal. When $|\mathbf{f}_{t_1}-\mathbf{f}_{t_2}|<1$, 
the two peaks merge together in the frequency domain resulting in a single peak at approximately $(A_2\mathbf{f}_{t_2}+A_1\mathbf{f}_{t_1})/(A_2+A_1)$, where $A_1$ and $A_2$ are the amplitude of the direct path and the total amplitude of the residual indirect paths. Hence, the frequency error is $(A_2\mathbf{f}_{t_2}+A_1\mathbf{f}_{t_1})/(A_2+A_1)-\mathbf{f}_{t_1}$ which is equivalent to a distance error given by,  
$
d_e^{(peak)}=(\frac{A_2\mathbf{f}_{t_2}+A_1\mathbf{f}_{t_1}}{A_2+A_1}-\mathbf{f}_{t_1})\frac{c}{B} = {(\mathbf{f}_{t_2}-\mathbf{f}_{t_1})}/{(1+\frac{A_1}{A_2})} \frac{c}{B}
$. {\it This error increases linearly with  $\mathbf{f}_{t_2}-\mathbf{f}_{t_1}$ and  proportionally increases with $\frac{A_2}{A_1}$.}}

\red{{\it Our method.}
In contrast, the error in the phase of the FMCW signal is significantly smaller. To see this, let us assume that the amplitude of the residual indirect paths after filters have a lower amplitude than the direct path.  As shown in Fig.~\ref{fig:phaseerror}, the complex representation of the direct path is represented by the blue vector while that of the sum of indirect paths is represented by the red vector. The sum of the two vectors is the resulting signal at the receiver which is represented by the green vector. The maximum phase error occurs when the red vector is perpendicular to the green vector and this corresponds to a phase error of $sin^{-1}(\frac{A_2}{A_1})$. {\it The key observation is that this error does not depend on $\mathbf{f}_{t_2}-\mathbf{f}_{t_1}$ and  increases much slower at  $sin^{-1}(\frac{A_2}{A_1})$. }
}

\red{
Fig.~\ref{fig:errcompare} shows that distance error as a function of $A_2/A_1$ and $|\mathbf{f}_{t_2}-\mathbf{f}_{t_1}|$ for both traditional peak estimation techniques as well as our FMCW phase method. The plots show that the distance errors using peak estimation is severely affected by the time-of-flight of the indirect paths. In contrast,  the distance error using our FMCW phase technique is around 10x lower. We  describe our algorithm and its theoretical analysis in more detail in the Appendix.}

\begin{figure}
\begin{minipage}{.23\textwidth}\centering
\includegraphics[width=0.7\textwidth]{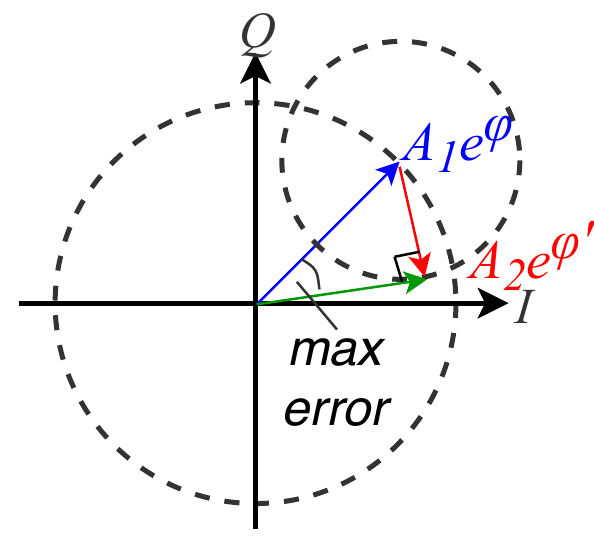}
\vskip -0.15in
\caption{Phase error when one indirect path (red vector) is combined with the direct path (blue vector).}
\vskip -0.15in
\label{fig:phaseerror}
\end{minipage}%
\hskip 0.1in
\begin{minipage}{0.23\textwidth}\centering
\includegraphics[width=1.1\textwidth]{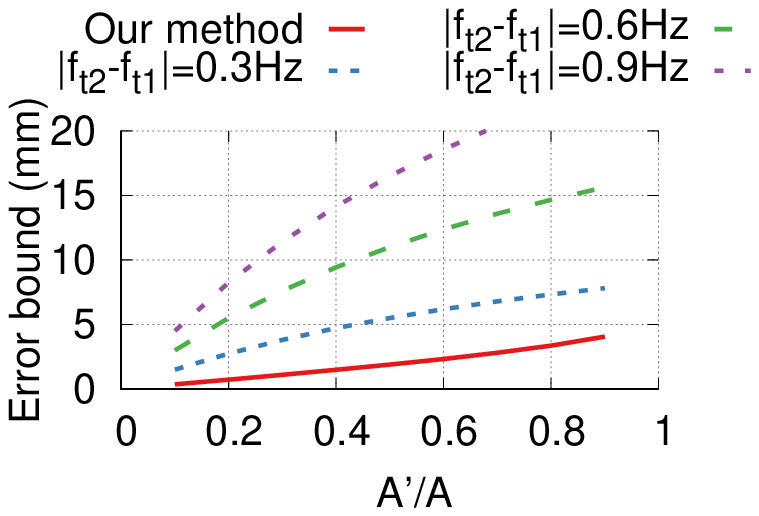}
\vskip -0.15in
\caption{Error comparison of FMCW peak versus our FMCW phase method.}
\vskip -0.15in
\label{fig:errcompare}
\end{minipage}
\end{figure}

\subsection{3D tracking from 1D locations}

The above FMCW phase technique allows us to achieve sub-mm resolution in estimating 1D distances. To achieve 3D tracking we use information from multiple microphones to perform triangulation.  We use multiple microphones instead of speakers to reduce the power consumption as well as to eliminate the complexity of multiplexing the multiple speakers. Since our 1D resolution is high we can also reduce the separation between the microphones  while achieving a good 3D accuracy. Specifically, we place four microphones at four corners of a rectangle.  We have two pairs of microphones in the vertical position and the other two pairs in the horizontal position. Thus, by computing the intersection of all the resulting 1D positions, we can compute the 3D location. 

We note that the accuracy of triangulation is dependent on the distance from the microphone array as well as the separation between the microphones. Specifically, as the distance from the microphone array increases, the resulting 3D accuracy become worse. Similarly, as the separation between microphones increase the 3D accuracy improves, which is why prior work uses a microphone separation of 90~cm~\cite{mao2016cat}. In our solution, since we already achieve sub-mm 1D resolution, we can reduce the separation between microphones to fit the form-factor of VR headset and still achieve good 3D tracking accuracies upto 2~m.  To improve the accuracy and reduce jitters at larger distances, we average the 3D distance measurements within each 10ms duration. Incorporating the 15ms latency of the band-pass filter, we get one distance value every 25ms or a frame rate of 40 frames per second.


\subsection{Addressing practical issues} We describe the practical issues in designing \name.

{\it 1) Phase ambiguity.} Any phase tracking algorithm has to address the problem of phase ambiguity. 
Specifically, we can only extract the phase modulo $2\pi$ from the demodulated chirp ($\hat{\phi}(t)=\mathbf{\phi}(t)\text{ mod }2\pi$).
This leads to two problems: a) how to detect any modulo $2\pi$ shifts during a single chirp; and b) how to estimate the initial $2N\pi$ phase offset, \textit{i.e.,} $\mathbf{\phi}(0)=2N\pi+\hat{\phi}(0)$ at the beginning of each chirp.

\red{Because of the band-pass filter, adjacent samples does not have a phase difference of more than $\pi$. The phase error caused by residual indirect paths is bounded to $(-\pi/2, \pi/2)$.  Thus,} when the phase modulo $2\pi$ sees a sudden jump of more than $\pi$/$-\pi$ between adjacent samples at $t$ and $t-\delta t$, there is a modulo $2\pi$ jump at that time, which we can correct by adding or subtracting $2\pi$ to the computed phase.

To compute the initial $2N\pi$ phase offset at the beginning of each chirp, we use the estimated distance and speed from the end of the previous chirp. Instantaneous speed is computed by performing least square linear regression (which is a linear algorithm in 1D domain) over the distance values in a 10~ms window to reduce the effects of noise and residual multipath.

Specifically, for the $i+1$th received chirp, given the distance $d^{(i)}_{end}$ and speed $v^{(i)}_{end}$ estimated from the end of the $i$th chirp, we can infer the distance of the beginning of the current chirp  $\hat{d}^{(i+1)}_{start}=d^{(i)}_{end}+v^{(i)}_{end}\delta T$ where $\delta T$ is the gap between two chirps. 
We then find the $2N\pi$ offset in addition to the ambiguous initial phase $\hat{\phi}(0)$ that minimize the difference $|d^{(i+1)}(0, \hat{\phi}(0)+2N\pi)-\hat{d}^{(i+1)}_{start}|$. Note that this relaxes the constraints on speed imposed by prior work~\cite{mao2016cat} and instead has a constraint on acceleration. Specifically, prior work~\cite{mao2016cat} had a  constraint on the maximum speed of 1~m/s. In our algorithm, since each  $2\pi$ difference of the phase corresponds to around $2$~cm distance difference, any error smaller than $1$cm will not cause an erroneous $2N\pi$ estimate. The gap between two adjacent chirp is 5~ms and the delay of the band-pass filter is 15~ms.  Hence, as long as the acceleration does not exceed $\frac{1cm}{20ms \times 20ms}=25m/s^2$, our algorithm does not introduce erroneous phase offsets.
\begin{figure}[t!]
    \centering
    \includegraphics[width=0.45\textwidth]{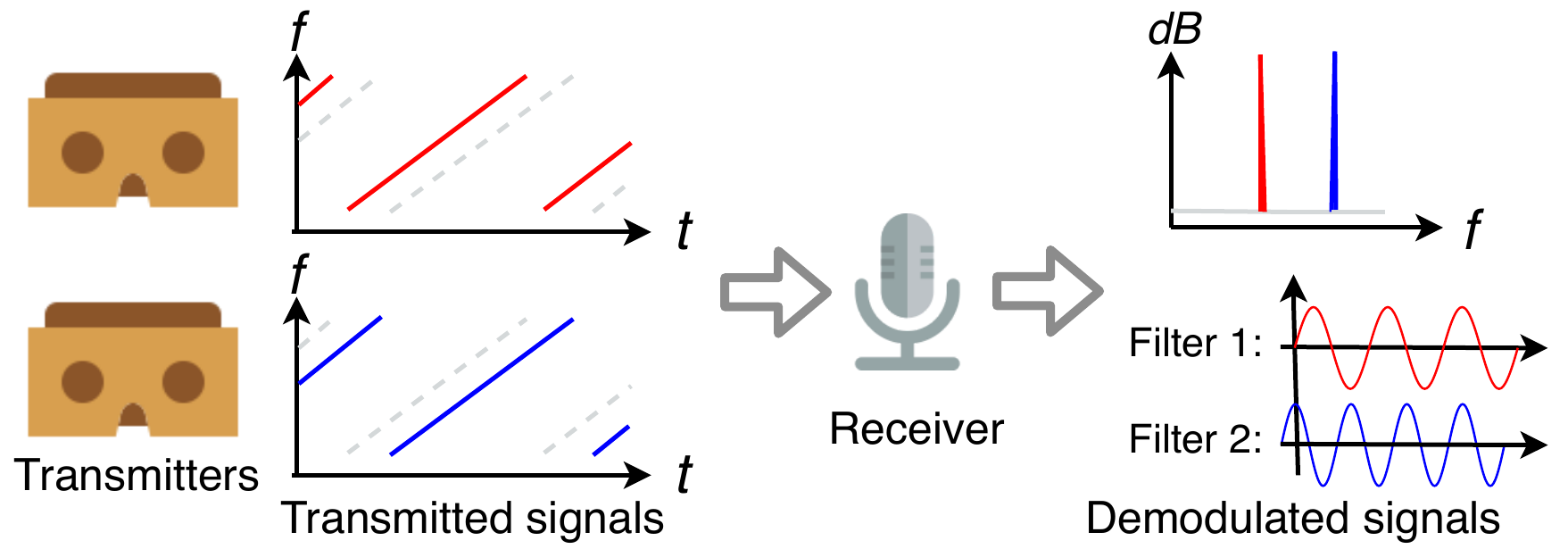}
    \vskip -0.15in
    \caption{Supporting concurrent transmissions using virtual time-of-arrival offsets at each VR headset.}
    \vskip -0.25in
    \label{fig:prototype}
\end{figure}

{\it 2) Clock synchronization.} Clock differences exist in practice which we need to calibrate to achieve tracking. Specifically, we need to calibrate for the initial phase as well as any drift due to clock differences.  To achieve this, at the beginning of the session,  the user touches the smartphone speaker with a microphone at the receiver.  The receiver meanwhile starts recording the chirp for five seconds and runs the above tracking algorithm. Using this setup, we use autocorrelation to determine a starting time for the chirp at the receiver. Because in this setup, at zero distance, the signal has a high SNR, there is no motion and little indirect path, the estimate of the distance from the peak of the FFT result, denoted by $D$, is accurate.  As a result,  we can find the initial $2N\pi$ phase offset from $D$. Specifically, we find the best $2N\pi$  which minimizes the differences of the two measurements $|D-d(0, \hat{\phi}(0)+2N\pi)|$. Finally, to address the clock drift, the receiver detects a constant drift in {the distance measurement} within the five seconds which is linear to the clock difference. We then compensate the clock difference by removing this drift for the following measurements at the receiver side.

{\it 3) Failure detection and recovery.} Our algorithms relies on continuous tracking. When tracking failure occurs, the subsequent measurements are also prone to errors.  In practice, failures occur due to  occlusions and  noise.  

While acoustic signal can penetrate some occlusions like fabrics, for other occlusions like wood and human limbs, refraction between different transmission mediums causes a dense multipath around the
direct path which is also greatly attenuated. Therefore, the above algorithm fails because it doesn't satisfy the premise that the direct path dominates the filtered demodulated signal. When such error happens during a chirp, it will cause fluctuations in phase. Thus, there would be multiple $2\pi$ phase tracking error during the chirp, leading to a larger than $2cm$ distance error at the end of the chirp. When the error happens between two chirps, it will lead to wrong $2N\pi$ estimate that causes larger than $~2cm$ error for the subsequent chirps. Similarly at longer ranges the signal attenuates which results in noisy phase measurements which can also lead to wrong $2N\pi$ estimates.

To detect these failures, we utilize the redundancy across the microphones. It is  unlikely that all the four microphone encounter the same extra phase error at the same time because of their different locations. Hence, if the measurements from some of the four microphones are outliers with at least $2cm$ measurement errors from the others, it indicates an error. 
When such a failure is detected, if the anomaly is only in one microphone, the receiver compensates the $2\pi$ offset until it is in the similar range of the other three microphones. If sustained failures occurs (which rarely happens), our algorithms fall back to the traditional peak estimation method for FMCW signals and notify the user.

\begin{figure}[t!]
    \centering
    \begin{subfigure}{0.178\textwidth}
    \includegraphics[width=1\textwidth]{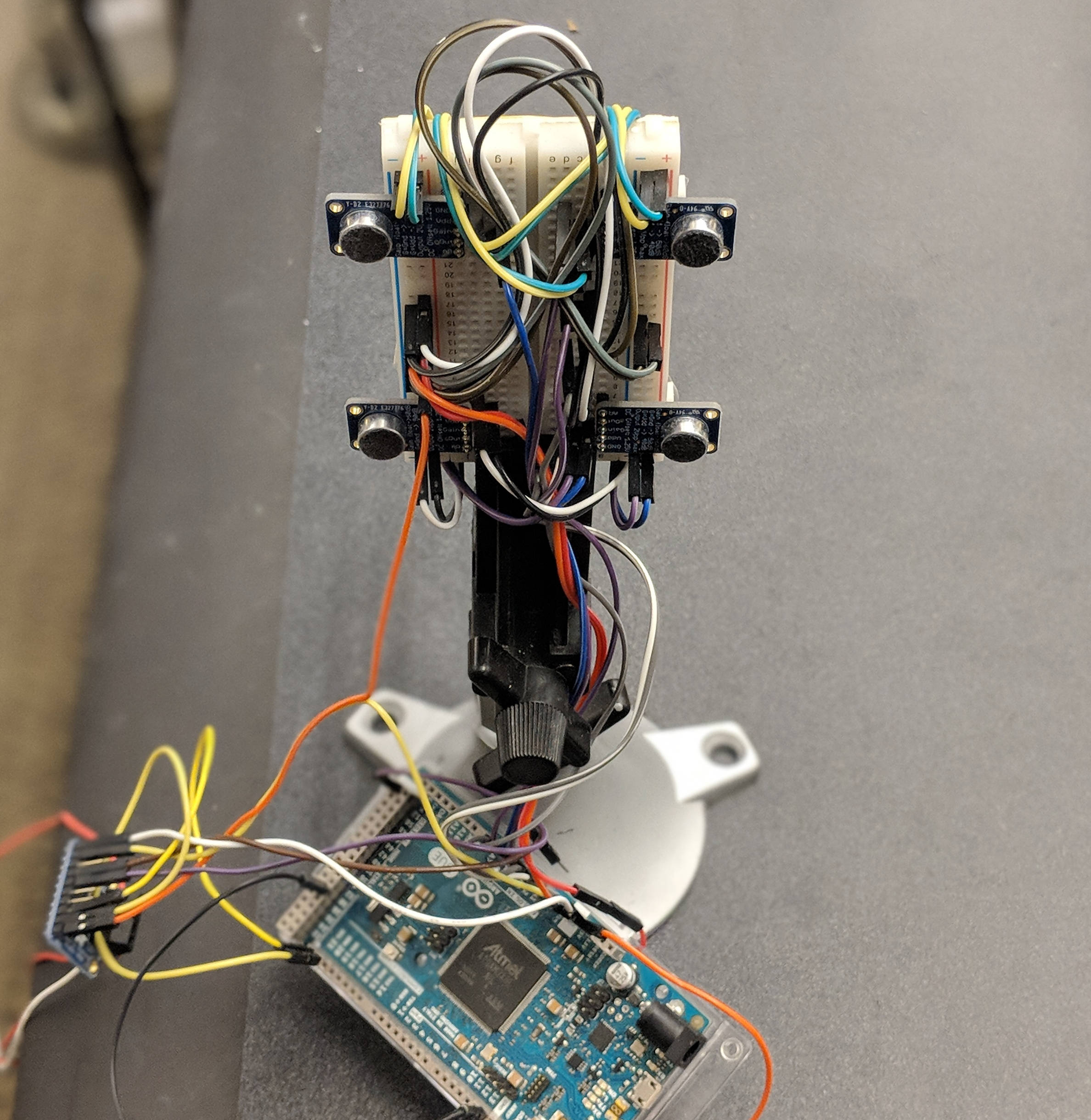}
    \caption{$6cm\times 5.35cm$}
    \end{subfigure}
    \hspace{2em}
    \begin{subfigure}{0.18\textwidth}
    \includegraphics[width=1\textwidth]{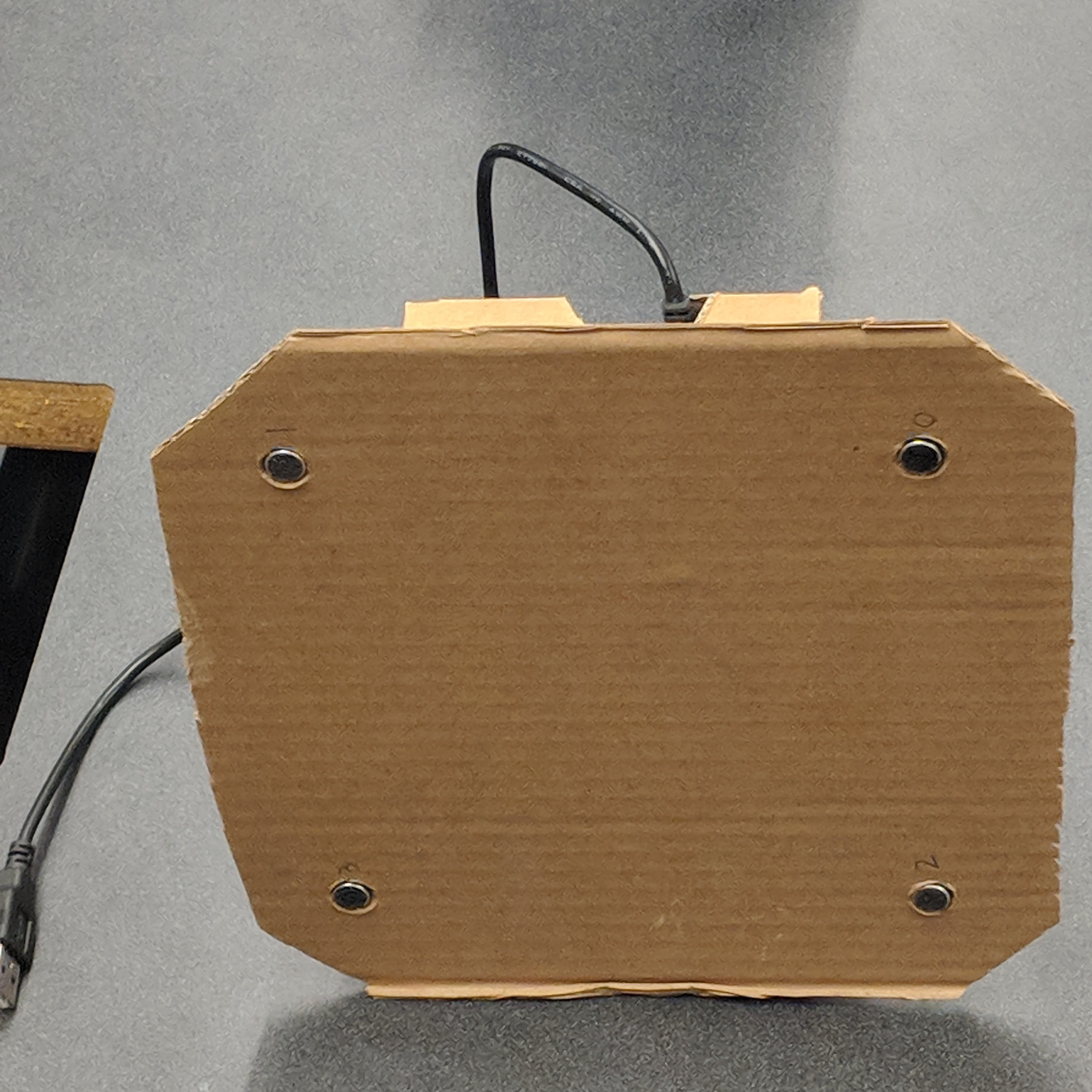}
    \caption{$15cm\times 15cm$}
    \end{subfigure}
    \vskip -0.15in
    \caption{Prototypes of MilliSonic microphone arrays.}
    \vskip -0.2in
    \label{fig:minisonic}
\end{figure}

\section{Tracking multiple devices}

The algorithm described above is unidirectional in that  signals can only \red{propagate} from the speaker to the microphones.  Because of this, MilliSonic can support tracking of an unlimited number of microphone arrays in the vicinity using a  one single smartphone speaker. Thus a single speaker can be used as a beacon to support tracking for multiple VR headsets that integrate our microphone array. 

On the other hand, tracking multiple smartphone-based VR headsets like the Google cardboard using a single microphone array is challenging since it involves transmissions from multiple smartphones.  Traditionally, wireless systems support multiple transmissions using either time-division multiplexing  or frequency-division multiplexing. In time-division multiplexing, since each smartphone speaker is only allowed to use a fraction of the time, it translates to a lower refresh rate that is inversely proportional to the number of smartphones. Using frequency-division multiplexing is challenging given the limited inaudible bandwidth on smartphones and since the accuracy depends on the bandwidth.

To achieve concurrent transmissions from all the smartphone speakers, we note that 
from Eq.~\ref{eq:fmcwdemod}, any two received FMCW paths with a time-of-arrival difference of $\delta t$, would lie in a different FFT bin. This indicates that two devices that have significantly different time-of-arrivals are at distant FFT bins and hence can be concurrently decoded.

We utilize this to support concurrent transmissions from multiple speakers. The challenge is that  two devices can have similar time-of-arrivals. To address this issue, we introduce \textit{virtual} time-of-arrival offsets at each device. Specifically, at the beginning, the  $N$ smartphones transmit FMCW chirps using time division. The receiver computes their time-of-arrivals using our algorithm, denoted by $t_d^{(i)}$ for the $i$th smartphone and sends back $\frac{iT}{2N}-t_d^{(i)}$ to each transmitter $i$, which is the virtual offset for transmitter $i$, using a Wi-Fi connection. The transmitter $i$ then intentionally delays its transmission by its virtual offset. The receiver picks these offsets to ensure that they are equally separated across all the FFT bins. This allows concurrent speaker transmissions.
 
Now at the receiver, there exist $N$ separate peaks evenly distributed in the frequency domain, which corresponds to $N$ evenly distributed time-of-arrivals, where the $i$th time-of-arrival is from the $i$th transmitter. 
The receiver can regard transmissions from other transmitters as multipath. Because of the orthogonality, they are filtered out by the band-pass filter at the first step. It can then track the phase of each of them using five different band-pass filters without losing accuracy nor frame rate.  After calculating the time-of-arrival of the signal from each speaker, it subtracts the virtual offset from it and obtains the final distance computation. 
\begin{figure}[!t]
\centering
\hskip -0.15in
\begin{subfigure}[b]{0.235\textwidth}
  \includegraphics[width=1.1\linewidth]{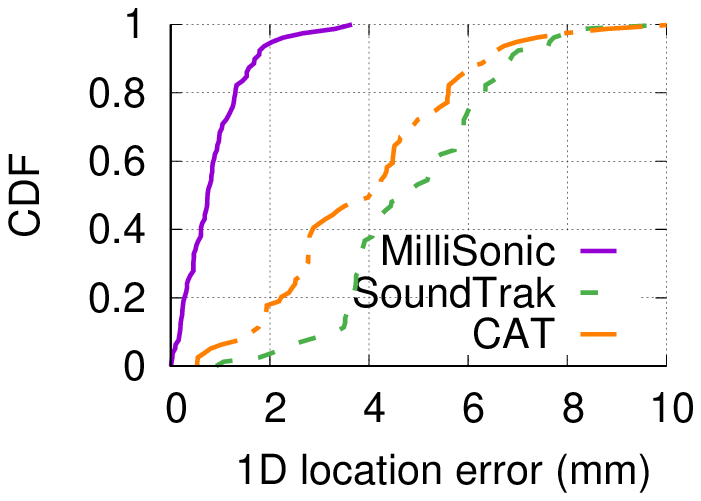}
  \caption{0-1m}
\end{subfigure}
\begin{subfigure}[b]{0.235\textwidth}
  \includegraphics[width=1.1\linewidth]{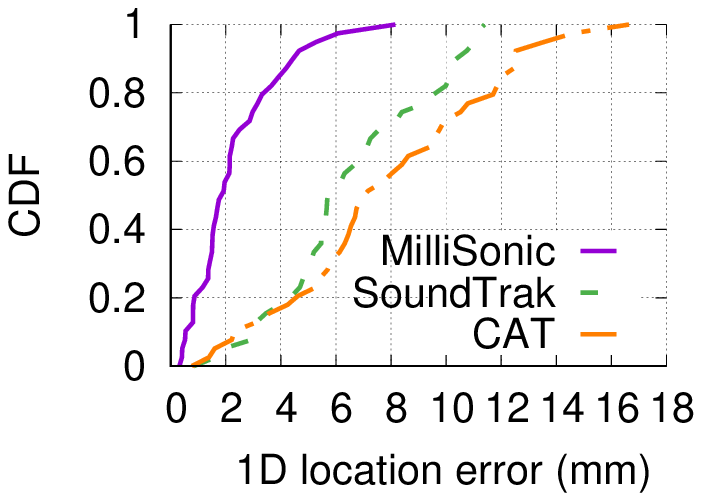}
  \caption{1-2m}
\end{subfigure}
\vskip -0.15in
  \caption{1D accuracy compared with CAT and SoundTrak.}\label{fig:rangecompare}
  \vskip -0.2in
\end{figure}

We note that because of motion,  over time, the time-of-arrivals for  multiple speakers can merge together. This would prevent the receiver from tracking all the devices concurrently. 
To prevent this, the receiver sends back a new set of virtual delays using Wi-Fi whenever the peaks between any two devices get close to each other in the FFT domain. \red{When the virtual delays get updated, which happens infrequently, there is an additional delay of at most one chirp duration (45 ms), divided by the number of transmitters.}

\section{Implementation}
We implement \name\ using  Android smartphones. We build an app that \red{emits} 45~ms 17.5-23.5~kHz FMCW acoustic chirps through the smartphone speaker. We tested it using Samsung Galaxy S6, Samsung Galaxy S9 and Samsung Galaxy S7 smartphones.  We build our microphone array using off-the-shelf electronic elements shown in Fig.~\ref{fig:prototype}. We use an Arduino Due connected to four MAX9814 Electret Microphone Amplifiers~\cite{max9814}. We attach the elements to a $20cm\times 20cm\times 3cm$ cardboard and place the four microphone on four corners of a $15cm\times 15cm$ square on one side of the cardboard. We also create a smaller {$6~cm\times 5.35~cm\times 3cm$} microphone array. We connect the Arduino to a Raspberry Pi 3 Model B+~\cite{raspberrypi} to process the recorded samples.   The software is implemented in the Scala programming language so that it can run on both a Raspberry Pi and a laptop without modification. It utilize multithreading to improve the performance. In our test, it requires 40ms and 9ms to process a single 45ms chirp on the Raspberry Pi and PC, respectively. Hence, it support real-time tracking on both platforms.

\section{Evaluation}

 We first evaluate the 1D and 3D tracking accuracy in a controlled \red{lab} environment. We then \red{recruited} ten participants  to evaluate the real-world performance of MilliSonic. 

\begin{figure}
    \begin{minipage}{.23\textwidth}\centering
    \includegraphics[width=1.1\textwidth]{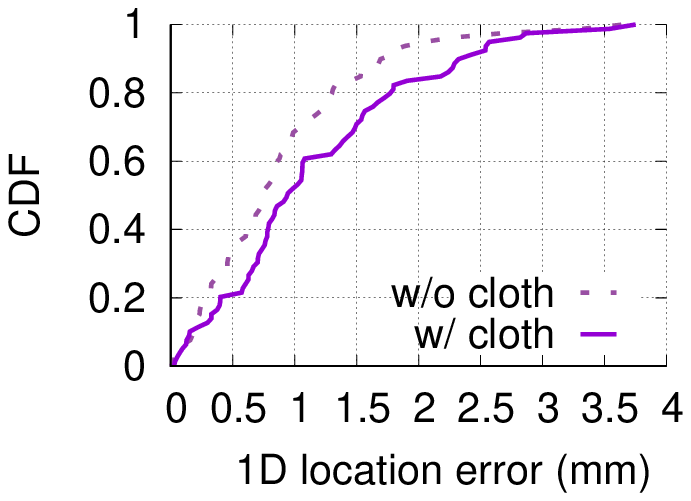}
    \caption{Impact of cloth as an occlusion.}
    \label{fig:occlusions}
    \end{minipage}%
    \hskip 0.1in
    \begin{minipage}{.23\textwidth}\centering
      \includegraphics[width=1.1\linewidth]{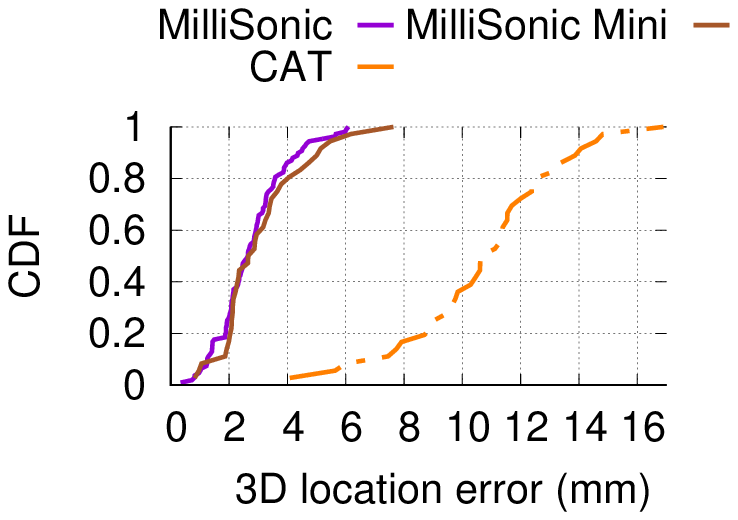}
      \captionof{figure}{3D localization accuracy.}
      \label{fig:3dcdf}
    \end{minipage}
    \vskip -0.15in
\end{figure}

\renewcommand{\subsection}[1]{\textbf{#1}}

\subsection{1D Localization Accuracy.}
To get an accurate ground truth, we use a linear actuator with a PhidgetStepper Bipolar Stepper Motor Controller~\cite{phidget} which has an movement resolution of $0.4\mu m$ to precisely control the location of the platform. We place a Galaxy S6 smartphone on the platform and place our microphone array on one end of the linear actuator. At each distance location, we repeat the algorithm ten times and record the measured distances.  We also implement  CAT~\cite{mao2016cat} and SoundTrak~\cite{zhang2017soundtrak}. CAT combines FMCW with Doppler effect that is estimated using an additional carrier wave and SoundTrak uses phase tracking. To achieve a fair comparison, we implement CAT using the same $6kHz$ bandwidth for FMCW and an additional $16.5kHz$ carrier. We implement SoundTrak using a $20kHz$ carrier wave. We do not use IMU data for all three systems.

Fig~\ref{fig:rangecompare}(a) and (b) plot the CDF of the 1D errors for two different distance ranges. We show the results for \name, CAT as well as SoundTrak. The plots show that our system achieves a median accuracy of {0.7~mm} up to distances of 1~m. In comparison, the median accuracy was {4 and 4.8} for CAT and SoundTrak respectively. When the distance between the smartphone and the microphone array is between 1--2 m, the median accuracy was {1.74~mm, 6.89~mm and 5.68~mm} for \name, CAT and SoundTrak respectively.  This decrease in accuracy is expected since with increased distance the SNR of the acoustic signals reduces. We also note that at closer distances, the error is dominated by multipath which our algorithm is designed to disambiguate multipath accurately.

\begin{figure}[!t]
\begin{subfigure}{0.235\textwidth}
  \includegraphics[width=1.1\textwidth]{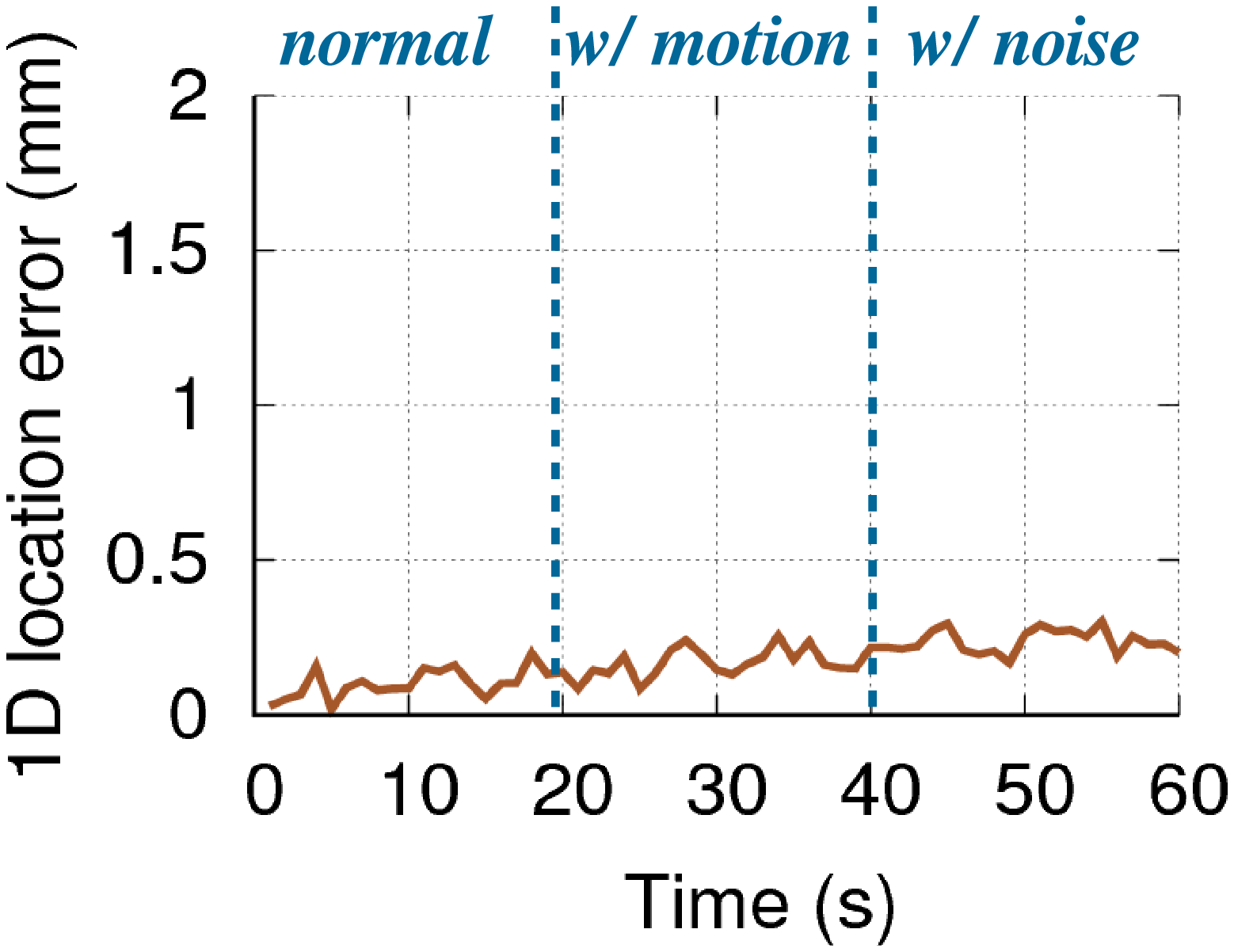}
  \caption{Motion and noise}\label{fig:surrounding}
\end{subfigure}
\begin{subfigure}{0.235\textwidth}
  \includegraphics[width=1.1\textwidth]{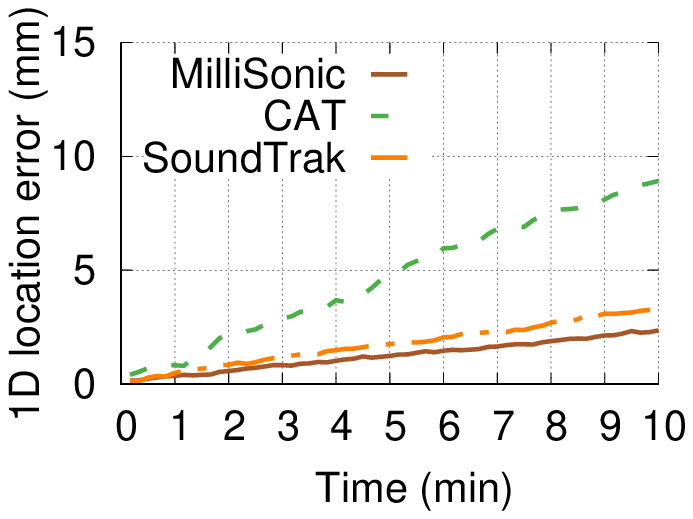}
  \caption{Long-term drift}\label{fig:drift}
\end{subfigure}
\vskip -0.15in
\caption{Effect of environmental motion, noise, and drift.}
\vskip -0.2in
\end{figure}

\subsection{Effect of environmental motion and noise.} We place the smartphone at $40cm$ on the linear actuator. We invite a participant to randomly move their body at a distance of $0.2m$ away from linear actuator. We also introduce acoustic noise by randomly pressing a keyboard and playing  pop music using another smartphone that is around $1m$ away from the linear actuator. 
Fig~\ref{fig:surrounding} shows the error. We can see that \name\ is resilient to random motion in the environment because of multipath resilience properties. Further, since we filter out the audible frequencies, music playing in the vicinity of our devices, does not affect its accuracy.

\subsection{Distance drift over time.} Tracking algorithms typically can have a drift in the computed distance over time.  We next, measure the drift in the location as measured by our system as a function of time. We also repeat the experiment for both CAT and SoundTrak. Specifically,
We place the smartphone at $40cm$ on the linear actuator for 10 minutes. We place the microphone array at the end of the actuator. We measure the distance as measured by each of these techniques over a duration of 10 minutes which we plot in Fig.~\ref{fig:drift}. SoundTrak and \name\ uses phase to precisely obtain the clock difference of the two devices, while CAT relies on \red{detecting the drift of peak frequencies}, which results in a larger drift. \red{With  a few millimeter drift at 10 minutes, \name\ has better stability than state-of-the-art acoustic tracking systems.} 

\red{\subsection{Effect of Environments.}} \red{To verify the robustness to different environments, we additionally evaluate the 1D accuracy in a) an anechoic chamber; b) a ~$200m^2$ lobby; and c) an outdoor open balcony; the median error was 0.75mm, 1.11mm and 0.94mm, respectively, at a distance of 0.6m.  }

\subsection{Tracking through occlusions.} Unlike optical signals, acoustic signals can traverse through occlusions like cloth. To evaluate this, we place the smartphone on a linear actuator and change its location between 0 to 1~m away from the microphone array. We place a cloth on the smartphone that occludes it from the microphone array. We then run our algorithm and compute the distance at each of the distance values. We repeat the experiments without the cloth covering the smartphone speaker. Fig.~\ref{fig:occlusions} plots the CDF of the distance error across all the tested locations both in the presence and absence of the cloth. The plots show that the median accuracy is {0.74~mm and 0.95~mm} in the two scenarios, showing that \name\ can track devices through cloth. This is beneficial when the phone is in the pocket and the microphone array is tracking its location through the fabric.

\subsection{3D Localization Accuracy}
Next, we measure the 3D localization accuracy of \name. To do this we create a working area of  $0.6m\times 0.6m \times 0.4m$. We then print a grid of fixed points  onto a $0.6m\times 0.6m$ wood substrate. We place the receiver on one side of the substrate, and place the smartphone's speaker at each of the points on the substrate. We also change the height of the substrate across the working area  to test the accuracy along the axis perpendicular to the substrate. To compare with prior designs, we run the same implementation of CAT as in our 1D experiments. Note that while CAT~\cite{mao2016cat} uses a separation of 90~cm, we still use $15cm$ microphone separation for CAT. This allows us to perform a head-to-head comparison as well as evaluate the feasibility of using a small  microphone array.


\begin{figure*}[!htb]
    \centering
    \begin{subfigure}{0.24\textwidth}
    \includegraphics[width=1\textwidth]{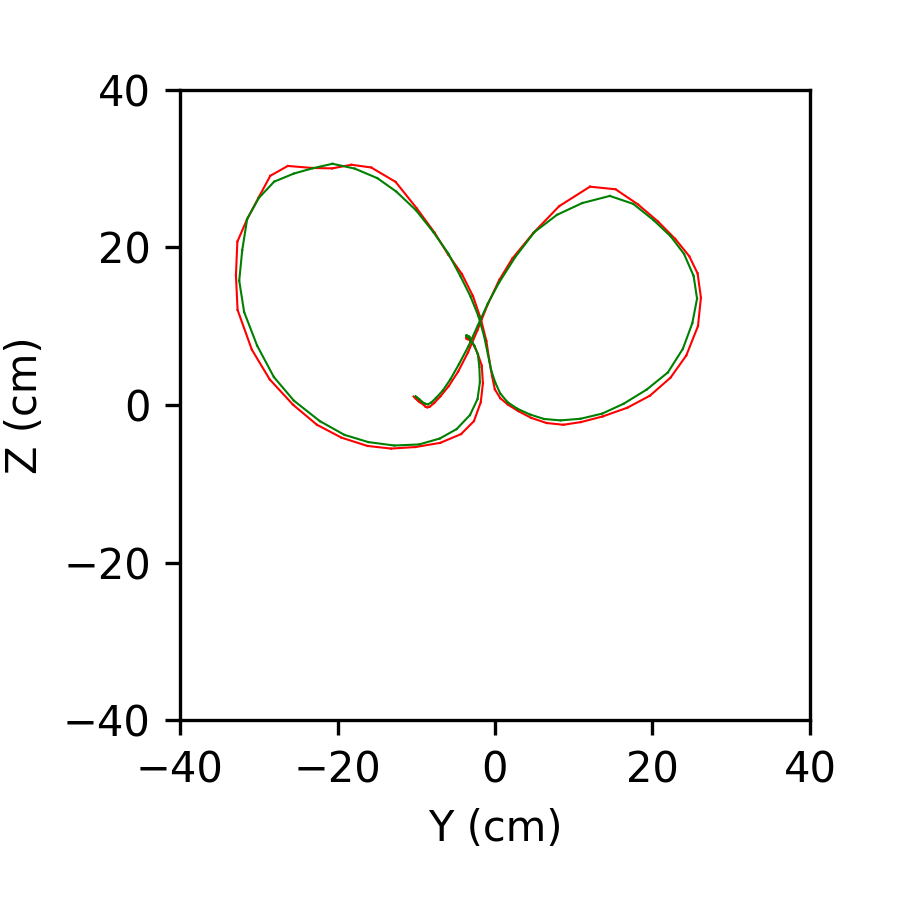}
    \caption{\textit{Infinity}}
    \end{subfigure}
    \begin{subfigure}{0.24\textwidth}
    \includegraphics[width=1\textwidth]{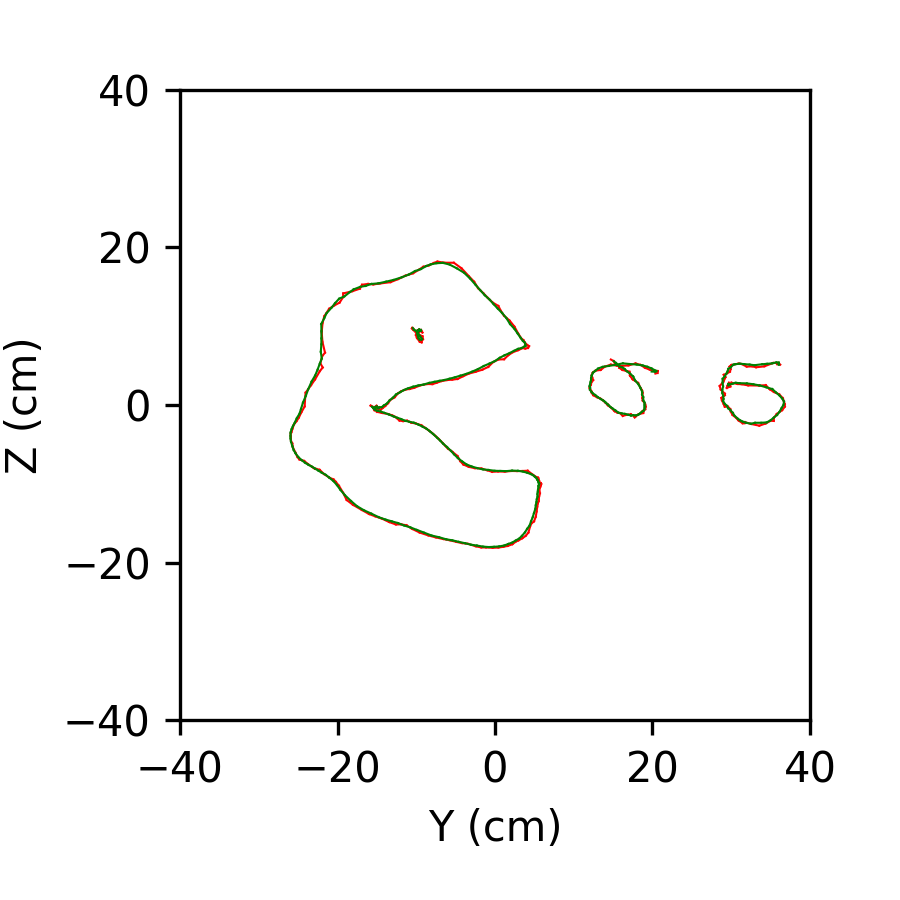}
    \caption{\textit{Pac-Man}}
    \end{subfigure}
    \begin{subfigure}{0.24\textwidth}
    \includegraphics[width=1\textwidth]{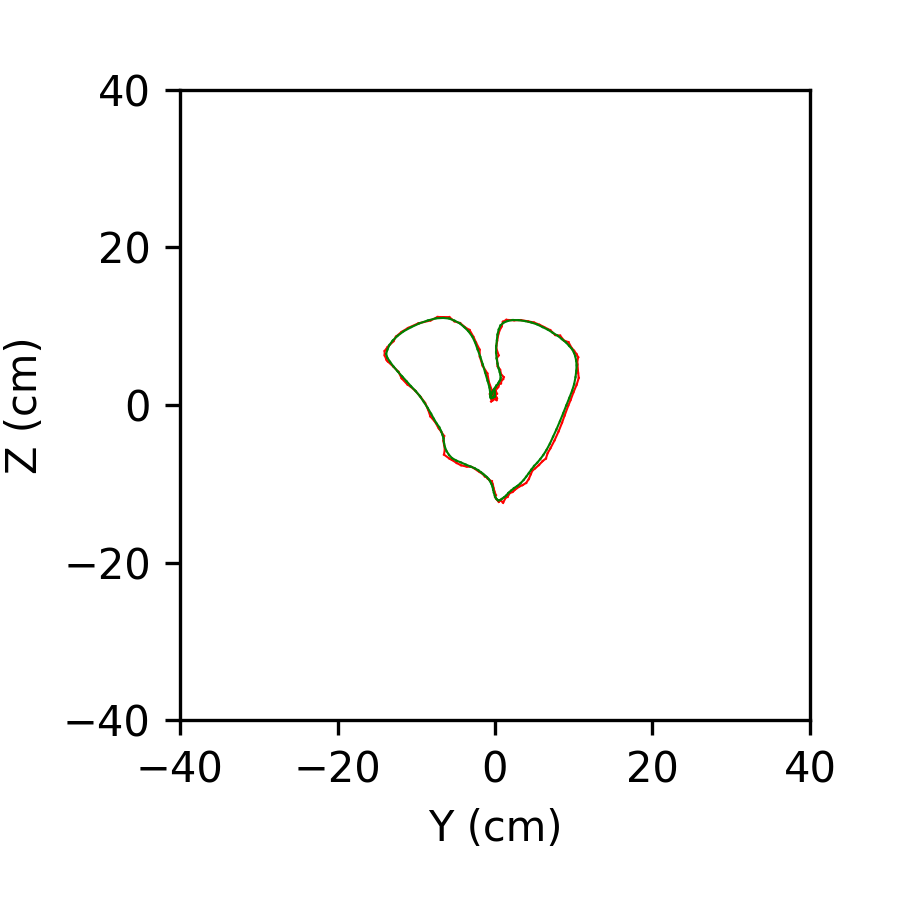}
    \caption{\textit{Heart}}
    \end{subfigure}
    \begin{subfigure}{0.24\textwidth}
    \includegraphics[width=1\textwidth]{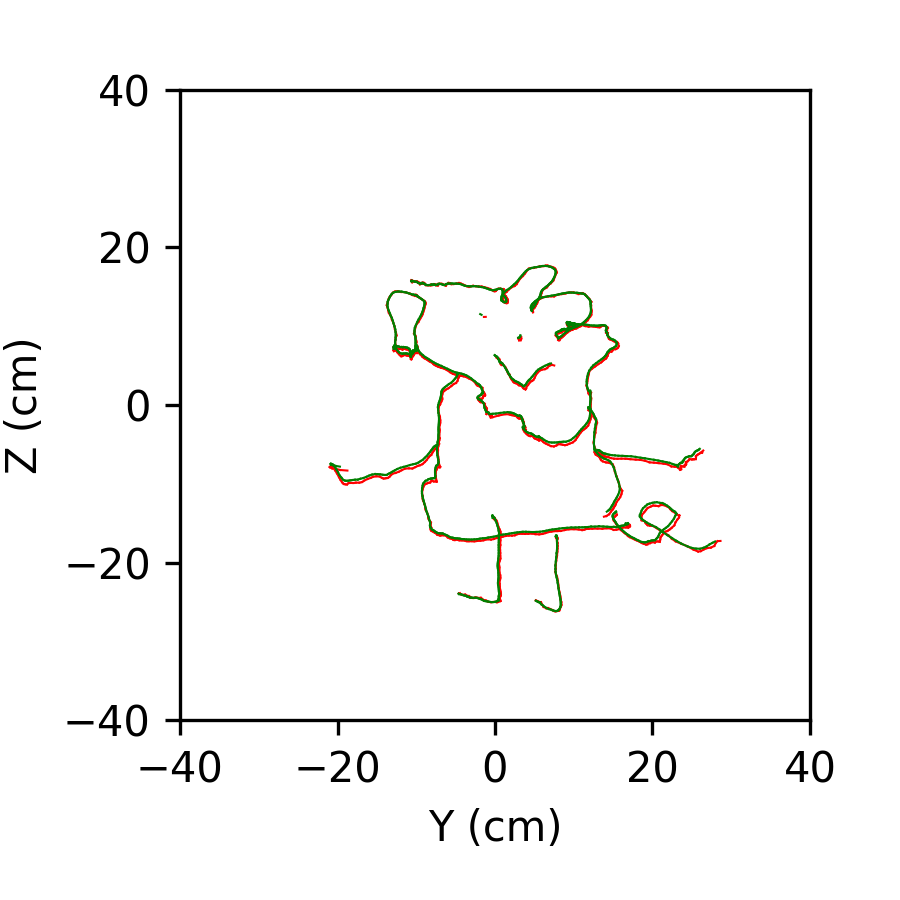}
    \caption{\textit{Peppa Pig}}
    \end{subfigure}
    \vskip -0.15in
    \caption{Sample drawings by  participants. Green and red traces are captured by HTC Lighthouse and MilliSonic  respectively.}
    \vskip -0.15in
    \label{fig:samples}
\end{figure*}

\begin{figure}
\vskip -0.1in
    \includegraphics[width=0.4\textwidth]{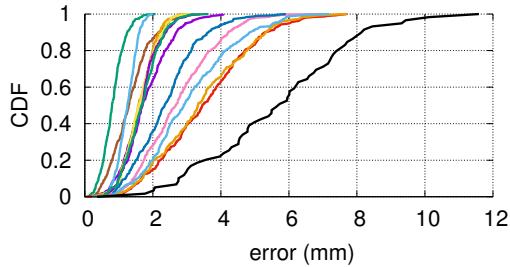}
    \vskip -0.15in
  \caption{The CDF of absolute 3D error across participants. \red{The black curve corresponds to the \textit{Infinity} in Fig.~\ref{fig:samples}}.}\label{fig:freemotionacc}
  \vskip -0.2in
\end{figure}

Fig~\ref{fig:3dcdf} shows the CDF of 3D location errors for \name\  and CAT in a working area across all the tested locations in our working area. The plots show that \name\ achieves a median 3D  accuracy of {2.6~mm} while CAT has a 3D accuracy of {10.6~mm}. The larger errors for CAT is expected since it is designed for microphone/speaker separations of 90~cm.

To understand the limits of the microphone separation, we further reduce the  microphone separation to $5.35cm$ using a breadboard hardware prototype as shown in Fig~\ref{fig:prototype}. This reduces the dimensions of the microphone array to approximately $6cm\times 6cm\times 3cm$. We show the 3D error results in Fig.~\ref{fig:3dcdf} labelled as \name\ Mini. We can see that there is little accuracy degradation. This shows that \name\ can enable a portable beacon design that uses microphone arrays to track smartphone based Google cardboard VR systems. Similarly, given these dimensions, the microphone array can be integrated into a VR headset which can then be tracked in 3D using a commodity smartphone as a beacon.

\subsection{Free Motion Tracking with Participants.} We build a simple \textit{draw-over-the-air} interface based on \name. We put our microphone array hardware on the table to act as the beacon. We implement a software app on Android platforms where participants can move the smartphone and touch the screen to draw 2D images on the $y-z$ plane over the air. Meanwhile, the strokes are rendered on an external screen in real-time.  We use a Samsung S6 smartphone for this study.  We compare \name\ to a HTC Vive Controller which is tracked using the HTC Lighthouse positioning system~\cite{htcvive}. Specifically we put two Lighthouse base stations on two tables with a distance of 2.5m. We attach the HTC Vive controller to the smartphone using tape and use the HTC Lighthouse positioning system to track its motion. Since the Lighthouse positioning system has an accuracy of around 1mm~\cite{niehorster2017accuracy}, we still use it as the ground truth. 

We recruit ten participants ({2 female and 8 male}) between the ages of {22-29} to draw on the air using \name. None of them were provided any monetary benefits. The participants were free to draw whatever they like and see the motion on the screen in real-time. 
We added a \textit{draw} button on the screen, so that when a user pushes the button, the app uses TCP to send the action to another server which records the traces and renders them on the screen in real-time. Each participant had to draw at least one figure of their choosing but could draw multiple figures if they wanted.  The participants in total drew 14 images. Fig.~\ref{fig:samples} shows five samples and the corresponding ground truth captured by a HTC Lighthouse.

We compare \name's accuracy with the ground truth from the HTC Lighthouse system. Because of frame rate differences, we linearly interpolate the ground truth result, find the point at the ground truth that is nearest to each point in our tracking result, and compute their difference. We show the CDFs of 3D accuracy in Fig.~\ref{fig:freemotionacc} for each of the 14 drawings which show accurate tracking capabilities using acoustic signals. The outlier orange curve corresponds to the \textit{infinite} drawing in Fig.~\ref{fig:samples} which shows that the practical accuracies are high. \red{There were a few instances when a wrong $2N\pi$ phase offset was estimated in the phase ambiguity removal algorithm on one of the microphones. This was however detected and successfully recovered by our failure recovery mechanism and did not affect the following chirp.}

We also measure the free motion speed distribution, acceleration distribution and distance distribution across the participants, which we plot in a Fig.~\ref{fig:userstudystat}. We  see  a  range of speeds and distances  during this user study. We also note that the maximum acceleration was $21~m/s^2$ with only 1 occurrence which was below our $25m/s^2$ limit.


\begin{figure*}[t!]
  \includegraphics[width=0.33\linewidth]{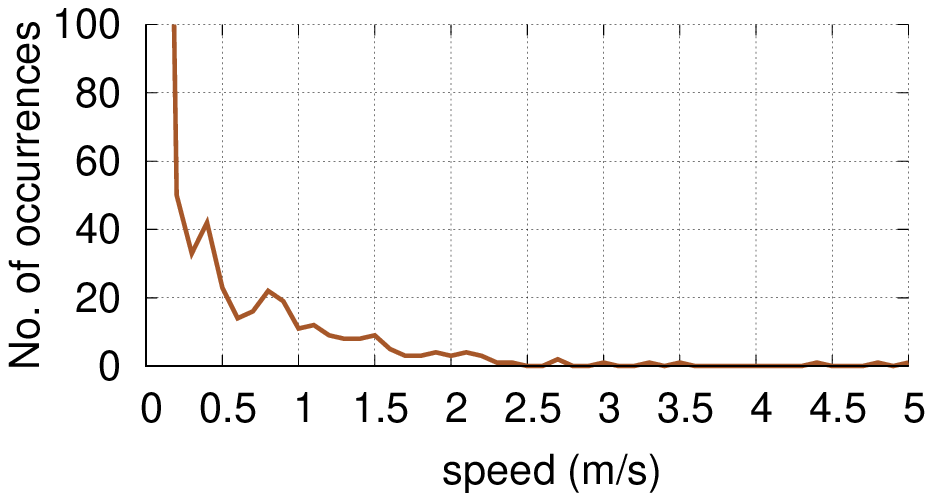}
  \includegraphics[width=0.33\linewidth]{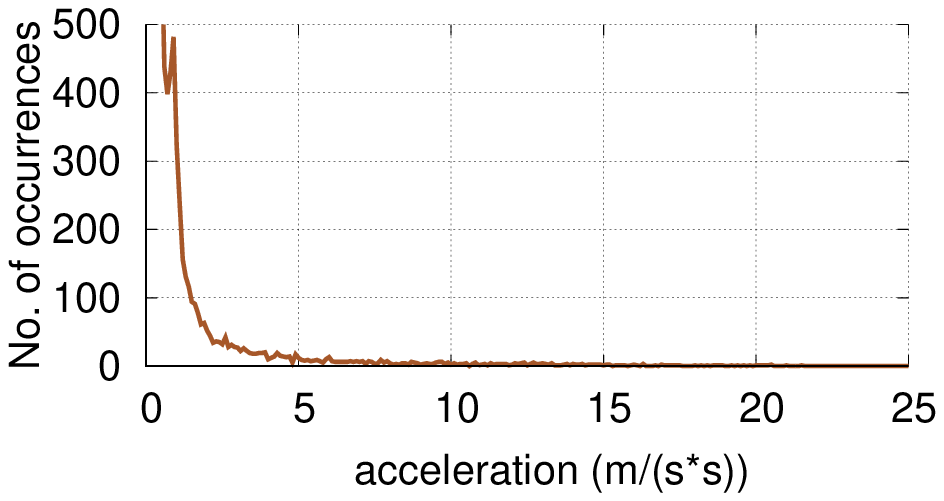}
  \includegraphics[width=0.33\linewidth]{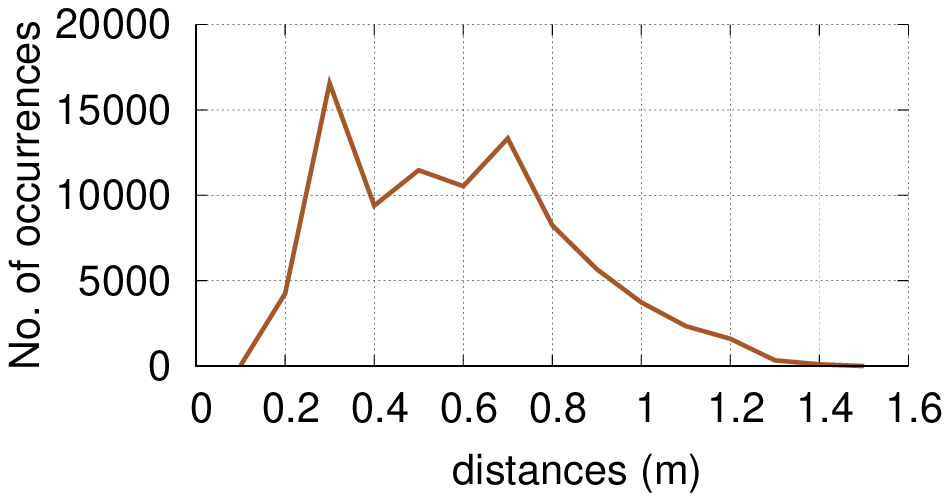}
  \vskip -0.15in
  \caption{Speed, acceleration and distance distribution during the user study.\\}\label{fig:userstudystat}
  \vskip -0.15in
\end{figure*}

\begin{figure}[b]
\vskip -0.25in
    \centering
    \includegraphics[width=0.4\textwidth]{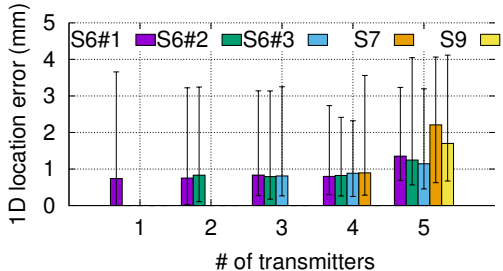}
    \vskip -0.15in
    \caption{Tracking error with concurrent smartphones.}
    \vskip -0.1in
    \label{fig:concurrent}
\end{figure}

\subsection{Enabling concurrent transmissions.} Finally, to evaluate concurrent transmissions with MilliSonic, we use five smartphones (3 Galaxy S6, 1 Galaxy S7, 1 Galaxy S9) as transmitters and one single microphone array to track all of them.  We use the same experimental setup as the 1D tracking, but place all five smartphones on  the linear actuator platform. We repeat experiments with different number of concurrent smartphones ranging from one to five. Fig.~\ref{fig:concurrent} shows the 1D tracking error of each of the smartphones in the range of 0-1m with different number of concurrent smartphones.  
We see that our system can support up to 4 concurrent transmissions without affecting the accuracy. With five concurrent smartphones, nearby peaks start to interfere with each other, resulting a slightly worse accuracy. 

\section{Related Work}
Prior work can be categorized as follows.

{\it Tracking using IMUs.} Inertial measurement units (IMUs) are a frequently used hardware to enable device tracking. IMUs sense 3D linear acceleration, rotational rate and heading reference which can all be fused together~\cite{fourati2015heterogeneous}. Gaming controllers~\cite{switch, wii} as well as many low-end VR systems~\cite{gearvr,oculusgo,daydream} use IMUs to support motion tracking. However IMUs do not accurately provide absolute positioning information. This is because position requires double integration of acceleration, which introduces a large drift error~\cite{feliz2009pedestrian}. 

{\it Tracking using IR/visible light.} The HTC Vive VR~\cite{htcvive} system uses a laser  \textit{Lighthouse} beacon emitting coherent IR signal to localize the headset as well as the controllers. Here, a laser emitter sweeps coherent IR light spatially and the 3D location is computed using  the time it takes for the IR signals to hit the photo-diodes on the receivers. Incoherent, infrared and visible light from LEDs can also be used for localization by cameras  using specific colors. The Sony PlayStation VR (PSVR)~\cite{playstationmove} system uses special visible colors that are tracked by a standalone camera located in a fixed position. Oculus Rift\cite{oculusrift} VR system employs a separate IR camera. The headset and controller are marked with IR LED markers captured by the IR camera. Despite being accurate enough for VR/AR applications, these techniques work poorly in bright environments~\cite{psvrissue}. More importantly, they  require a dedicated beacon hardware. In contrast, our design can use a smartphone as a basestation for tracking the VR headset.

{\it Tracking using cameras.} Unlike previous methods, Simultaneous Localization and Mapping (SLAM) techniques have also been used to enable tracking without relying on any beacon infrastructure. Using SLAM, devices can locate  themselves solely based on the environment captured by its camera. AR systems such as Microsoft Hololens\cite{hololens} and Magic Leap One\cite{magicleap} headsets use SLAM to achieve such tracking capabilities.
SLAM performance however highly depends on the environment including light conditions and variety of visual features~\cite{montemerlo2002fastslam,taketomi2017visual}. Hence, it is not as robust as outside-in tracking methods. SLAM is also a computational intensive algorithm that often requires specialized hardware accelerators to support real-time tracking. As a result, SLAM is unlikely to be appropriate for tracking tiny controllers.

\begin{table*}[t]
\centering
\resizebox{\textwidth}{!}{
\begin{tabular}{|l|l|l|l|l|l|l|l|l|l|l|l|l|}
\hline
System                & Setup                                                                        & \begin{tabular}[c]{@{}l@{}}Ranging\\ technique\end{tabular}                     & \begin{tabular}[c]{@{}l@{}}Need\\ IMU\end{tabular} & Audible            & Dimension & Accuracy & Latency               & \begin{tabular}[c]{@{}l@{}}Refresh\\ rate\end{tabular} & Range               & \begin{tabular}[c]{@{}l@{}}Concurrent\\ transmission\end{tabular}  & \begin{tabular}[c]{@{}l@{}}Mic/speaker\\ Separation\end{tabular}          \\ \hline
BeepBeep              & Phone-Phone                                                                  & Autocorrelation                                                                 & N                                                  & Y                  & 1D        & 2cm      & 50ms                  & 20Hz                                                   & 12m                    & N                                             & -                                                                            \\ 
Swordfight            & Phone-Phone                                                                  & Autocorrelation                                                                 & N                                                  & Y                  & 1D        & 2cm      & 46ms                  & 12Hz                                                   & 3m                   & N                                                 & -                                                                             \\ 
CAT                   & Speaker-Phone                                                                & FMCW                                                                            & Y                                                  & N                  & 3D        & 9mm      & 40ms                  & 25Hz                                                   & 7m                   & N                                               & 90cm                                                                          \\ 
SoundTrak             & Speaker-Watch                                                                & Phase tracking                                                                  & N                                                  & Y                  & 3D        & 13mm     & 12ms                  & 86Hz                                                   & 20cm                 & N                                               & 4cm                                                                           \\ 
Sonoloc               & Phone-Phone                                                                  & Autocorrelation                                                                 & N                                                  & Y                  & 2D        & 6cm      & 3.2-48s               & -                                                      & 17m                  & N                                                 & -                                                                             \\ \hline
\multirow{2}{*}{\name} & \multirow{2}{*}{\begin{tabular}[c]{@{}l@{}}Microphone-\\ Phone\end{tabular}} & \multirow{2}{*}{\begin{tabular}[c]{@{}l@{}}FMCW+\\ Phase tracking\end{tabular}} & \multirow{2}{*}{N}                                 & \multirow{2}{*}{N} & 3D        & 2.6mm    & 25-40ms & \multirow{2}{*}{$\geq$40Hz}                    & \multirow{2}{*}{3m} & \multirow{2}{*}{Y}                               & \multirow{2}{*}{6-15cm}                                        \\ \cline{6-8}
                      &                                                                              &                                                                                 &                                                    &                    & 1D        & 0.6mm    &  15-30ms                     &                                                        &                                        &                                 &                                                                                   \\ \hline
\end{tabular}
}
\caption{Prior works on acoustic device tracking.}
\vskip -0.2in
\label{tab:acousticworks}
\end{table*}

{\it \red{Device tracking using acoustic signals}.} Table.~\ref{tab:acousticworks} shows recent work on acoustic localization and tracking. BeepBeep~\cite{peng2007beepbeep} and Swordfight~\cite{zhang2012swordfight} track 1D distances between phones but do not achieve 3D localization.  Sonoloc\cite{erdelyi2018sonoloc} realizes distributed localization and requires 10+ devices to achieve reasonable accuracies. Prior work~\cite{infocom2018} also achieves 2D tracking by assuming that there is no significant multipath. \red{ALPS~\cite{alps} and Tracko~\cite{tracko} achieve centimeter-level accuracy using a combination of Bluetooth and ultrasonic.} The closest to our work are CAT~\cite{mao2016cat} and SoundTrak~\cite{zhang2017soundtrak}. CAT achieves a median 3D error of 9mm using a combination of IMU sensor data and FMCW localization to address multipath. It requires a separation of 0.9m and 0.7m between its horizontal and vertical speaker pairs respectively. As a result, we cannot have a smartphone track the position of a VR headset, since both the devices have much smaller dimensions. SoundTrak\cite{zhang2017soundtrak} achieves an average 3D error of 1.3cm between a smart watch and a customized finger ring using phase tracking where the area of movement is limited to a $20cm\times16cm\times11cm$ space. Our work builds on this foundational work but is the first to 1) achieve sub-millimeter 1D resolution, 2)  do so without requiring large separation between microphones/speakers and 3) enable for the first time concurrent transmissions where all the acoustic devices transmit at the same time; thus allowing for high refresh rate in the presence of multiple trackers.

\red{{\it Device-free tracking using acoustic signals}.  VSkin~\cite{vskin} tracks gestures on the surface of mobile devices with a 2D accuracy of 3mm. Strata~\cite{strata}, LLAP~\cite{LLAP} and FingerIO~\cite{fingerio} track moving fingers in the proximity of a mobile device with a 2D accuracy of 1cm, 1.9cm and 1.2cm respectively. Toffee~\cite{toffee} localizes the direction of a touch around a mobile device of an angular error of 4.3\textdegree. While device-free finger tracking is challenging because of noisy measurements, it benefits from lack of  synchronization issues and relies on more strict multipath assumptions. This line of work however is complimentary to our work on acoustic device tracking. }

\section{Conclusion and Discussion}
We present \name, a novel system that pushes the limits of acoustic based motion tracking and localization.  We show for the first time how to achieve sub-mm 1D tracking and localization accuracies using acoustic signals on smartphones, in the presence of multipath. To achieve this, we introduce algorithms that use the phase of FMCW signals to disambiguate between multiple paths. We also enable multiple  smartphones to transmit concurrently using time-shifted FMCW acoustic signals and enable concurrent tracking  without sacrificing  accuracy or frame rate. 

\red{While this paper presents multiple benchmarks, user studies and evaluation in indoor and outdoor environments, more extensive evaluation is required to understand its behavior in various edge cases as well as in rooms with significant multipath that can adversely affect accuracy. Here, we discuss the limitations of our current system design. }


\red{First, we support simple occlusions such as fabric and paper, but do not support human limbs or the device itself. Additional algorithmic development is required to support these practical occlusion scenarios. } \red{Second, while our design has better drift characteristics than prior work on acoustic tracking, further work is required to make it comparable to optical based systems. One approach is to perform sensor fusion with IMU data and achieve better accuracy, \red{lower latency} and more resilience to clock drifts. This could also enable VR headset tracking while using a mobile beacon (i.e., smartphone) in the hand instead of placing it on a table. }

\red{Our current range is limited to 2~m. This is because the microphones in our array prototype are not optimized for performance and are not designed to have optimal response in the 17.5--23.5~kHz frequencies. Finally, we support upto 4--5 concurrent smartphone acoustic transmissions without affecting the frame rate per device. One way to increase the number of concurrent devices is to use longer chirps so as to support more time-shifted FMCW chirps that can be allocated to different smartphones. This however comes at the expense of the frame rate per device.}



\section{Appendix: 1D Tracking Details}
\red{Our 1D tracking algorithm has two main components.}

{{\it 1) Adaptive band-pass filter to remove distant multipath.}  For the first FMCW chirp, we extract the first peak of the demodulated signal in the frequency domain using an DFT similar to prior designs from Eq.~\ref{eq:fmcwdemod}. We then apply a \red{Finite Impulse Response (FIR)} filter that only leaves a narrow range of frequency bands around the peak. We adaptively set the delay of the FIR filter from the SNR of the acoustic signals. Specifically, when $SNR>10dB$, we use a 15ms delay; otherwise, we double the delay to $30ms$.}

{For subsequent FMCW chirps we no longer use the DFT to extract the peak frequency. Instead, for the $i+1$th FMCW chirp, we infer the new peak from the distance and speed estimated at the end of the $i$th chirp. We then apply the FIR filter around this new peak. Given the distance $d^{(i)}_{end}$ and speed $v^{(i)}_{end}$ estimated from the end of the $i$th chirp, we can infer the distance of the beginning of the current chirp  $\hat{d}^{(i+1)}_{start}=d^{(i)}_{end}+v^{(i)}_{end}\delta T$ where $\delta T$ is the gap between two chirps. We do this for two key reasons: a) our distance estimates are far more accurate than the peak of the DFT result; and b) unlike a DFT that is performed over a whole FMCW chirp, we do not require receiving a full FMCW chirp before processing, thus reducing the frame rate. }

{Finally, Doppler effects can blur the peak in the frequency domain. So, we adaptively increase the width of the pass band in the FIR filter when the speed estimate at the end of the previous chirp exceeds a given threshold. In our algorithm, we set the pass band width to $1Hz$ when the speed does not exceed $1m/s$; otherwise, we set the pass band width to $2Hz$. }

{{\it 2) Extracting distance from FMCW phase.} The above process eliminated all multipaths that have a much larger time-of-flight than the direct path. This leaves us with residual indirect paths around the direct path. \red{Thus, when there is no occlusion, the sum of the residual indirect paths has a lower amplitude than the direct path (confirmed empirically).} }

{To extract the distance from the phase value, we approximate the effect of residual multipath after filtering. From Eq.~\ref{eq:fmcwdemod}, we approximate the phase as, 
\begin{equation}
    \mathbf{\phi}(t) \approx -2\pi(\frac{B}{T}tt_d+f_0t_d-\frac{B}{2T}t_d^2)
    \label{eq:phasequadratic}
    \end{equation}
Where $t_d$ is the time of arrival of the direct path. The approximation assumes that we have already applied the dynamic filter to remove most multipath that has a much larger time-of-arrival distance than the direct path.  The above quadratic equation in $t_d(t,\phi(t))$ can be uniquely solved; the equation has two solutions but only one is in the range of the  FMCW chirp, $[0,T]$.  The distance $d(t, \phi(t))$ can then be computed as, $ct_d(t, \phi(t))$. We note the following about the distance error.}

\begin{lemma} Given the phase error bound of $sin^{-1}(\frac{A_2}{A_1})$ from Fig.~\ref{fig:errcompare}, the error in our distance estimate, $d(t)$ is upper bounded by $\frac{sin^{-1}(A_2/A_1)c}{2\pi(f_0-B/2)}$, where $f_0$ and $B$ are the FMCW parameters. 
\end{lemma}
\begin{proof} 
First we show that the function $\phi(t, t_d)$ in Eq.~\ref{eq:phasequadratic} is convex with respect to $t_d$, which is the time-of-arrival. This is because the first derivative is given by, 
$\frac{d\phi(t, t_d)}{dt_d}=-2\pi(\frac{B}{T}t+f_0-\frac{B}{T}t_d)<0$, when $t_d<T$. Further its second derivative
$\frac{d^2\phi(t, t_d)}{dt_d^2}=2\pi\frac{B}{T}>0$ resulting in a convex function.

Suppose an phase error $\phi_e$ would introduce a time-of-arrival error of $t_e$ because of multipath.  Without loss of generality, we assume $\phi_e>0$. we know that for any $t$ for a convex function, $\phi(t, t_d)>\phi(t, t_d+t_e)-\frac{d \phi(t, t_d+t_e)}{d t_d}t_e$.  Thus the error in the phase $\phi_e = \phi(t, t_d) - \phi(t, t_d+t_e)$ can we written as,
$
\phi_e>-\frac{d\phi(t, t_d+t_e)}{dt_d}t_e
$.
This  can be rewritten as, $
t_e<\frac{\phi_e}{-\frac{d\phi(t, t_d+t_e)}{dt_d}} =\frac{\phi_e}{2\pi(\frac{B}{T}t+f_0-\frac{B}{T}(t_d+t_e))}$.  The upper bound for this  equation occurs when $t=0$ and $t_d+t_e$ is maximum. Since the maximum delay permitted by our FMCW signal is half its duration, this occurs when $t_d+t_e=\frac{T}{2}$. First from Lemma 2.1, we know that $\phi_e < sin^{-1}(\frac{A_2}{A_1})$. Thus the above equation is upper bounded as:
$t_e < \frac{sin^{-1}(\frac{A_2}{A_1})}{2\pi(f_0-\frac{B}{2})}$. Thus,  
$d_e^{(phase)}<\frac{sin^{-1}(\frac{A_2}{A_1})c}{2\pi(f_0-\frac{B}{2})}$.

\end{proof}

\bibliographystyle{acm}
\bibliography{main} 
\balance
\end{document}